\documentclass[11pt]{article}

\usepackage[a4paper,margin=1in]{geometry}
\usepackage{setspace}
\usepackage{amsmath,amssymb,amsthm}
\usepackage{mathtools}
\usepackage{bm}

\usepackage[T1]{fontenc}
\usepackage[utf8]{inputenc}

\usepackage{hyperref}
\hypersetup{
  colorlinks=true,
  linkcolor=blue,
  citecolor=blue,
  urlcolor=blue
}

\usepackage[round,authoryear]{natbib}
\newtheorem{theorem}{Theorem}
\newtheorem{assumption}{Assumption}
\newtheorem{definition}{Definition}

\newtheorem{lemma}{Lemma}
\newtheorem{corollary}{Corollary}

\DeclareMathOperator{\logit}{logit}

\usepackage{authblk} 

\title{The General Formulation of Loss-Based Priors for Parameter Spaces}
\author[1,2,3]{Cristiano Villa}
\affil[1]{Duke Kunshan University, China}
\affil[2]{Duke University, US}
\affil[3]{Newcastle University, UK}
\date{} 

\begin{document}

\maketitle

\begin{abstract}
Loss-based priors assign probability mass to parameter values according to the
inferential loss incurred when they are excluded from the parameter space, and
provide a general solution for discrete parameters. Extending this idea to
continuous settings is challenging, as the exclusion of a single point induces
no loss. We propose a neighbourhood-exclusion framework in which inferential
loss is defined by removing a local region around each parameter value. Under
standard regularity conditions, this yields a class of prior distributions
driven by the local geometry of the Kullback--Leibler divergence. In one
dimension, the resulting prior coincides with Jeffreys' prior, while in higher
dimensions it leads to a family of priors indexed by the geometry of the
exclusion region. The proposed formulation provides a unified extension of
loss-based priors and offers a geometric interpretation of objective prior
construction beyond isotropic settings.
\end{abstract}

\section{Introduction}\label{sc_introduction}

Prior specification is a central component of Bayesian inference, particularly in
settings where little or no prior information is available. Objective Bayes seeks
to construct priors according to formal principles, often derived from the
sampling model, such as invariance under reparametrisation, \citep{Jeffreys1961} maximisation of
missing information \citep{BergerBernardoSun2009}, or penalisation of model complexity
\citep{Simpsonetal2017}. For a discussion in support of Objective Bayes, refer to \cite{Berger2006}, while overviews of objective methods can be found in \cite{KassVasserman1996} or in \cite{ConsonniFouskakisLiseoNtzoufras2018}, for example. Finally, for a novel class of objective priors that depend solely on the parameter space, refer to \cite{Leisenetal2020}, \cite{VillaWalker2017} and to \cite{Antonianoetal2024}.

Despite substantial progress, most objective prior constructions are inherently
tailored to continuous parameter spaces. In particular, methods based on Fisher
information or asymptotic information measures do not extend naturally to discrete
parameters \citep{Rissanen1983,LindsayRoeder1987,BargerBunge2008,BBS2012}. Conversely, approaches developed for discrete settings are typically
not formulated in a way that admits a direct extension to continuous spaces.

A general framework for discrete parameter spaces was proposed in
\citet{VillaWalker2015}, where prior mass is assigned according to the inferential
loss incurred by removing a parameter value from the space. This construction,
based on Kullback--Leibler divergence, yields priors that depend only on the model
and the parameter space, and has been successfully applied in a range of discrete
problems \citep{VillaWalker2015scan,VillaLee2020,Hinoveanuetal2019,Grazianetal2020,Serafinietal2024}.

In this paper we extend the loss-based framework to continuous parameter spaces.
The key idea is to replace point exclusion with the exclusion of a neighbourhood
around each parameter value. This modification restores a non-trivial notion of
inferential loss and leads to a class of prior distributions indexed by the
geometry of the exclusion region.

Under standard regularity conditions, the resulting priors are determined by the
local behaviour of the Kullback--Leibler divergence. In one dimension, the
construction recovers Jeffreys' prior. In higher dimensions, it yields a class of
priors indexed by a positive definite matrix $A(\theta)$, which governs the
geometry of neighbourhood exclusion. Different choices of $A(\theta)$ lead to
distinct but interpretable priors, including both isotropic (volume-based) and
anisotropic (direction-sensitive) constructions.

The remainder of the paper is organised as follows. Section
\ref{sc_background} reviews the loss-based framework for discrete parameter
spaces. Section \ref{sc_neighbourhoodexclusion} introduces neighbourhood
exclusion for continuous parameters. Section \ref{sc_lossbasedprior} derives the
resulting class of priors. Section 5 discusses their main properties, and Section
6 concludes.

% ...
\section{Background: The loss-based framework}\label{sc_background}

We briefly review the loss-based prior construction for discrete parameter
spaces introduced in \citet{VillaWalker2015}, which forms the basis for the
extension developed in this paper.

Let $\theta\in\Theta$ be a discrete parameter and let $f(\cdot\mid\theta)$
denote the statistical model. Suppose that $\theta_0$ is the true parameter
value. If $\theta_0$ is removed from the parameter space, inference must proceed
using an alternative value $\theta_j\neq\theta_0$. The associated inferential
loss can be quantified through the Kullback--Leibler divergence between the true
model and the closest remaining alternative \citep{Berk1966}. Thus, the \emph{worth} of $\theta_0$ is defined as
\begin{equation}\label{eq:worthdiscrete_trim}
u(\theta_0)
=
\min_{\theta_j\neq\theta_0}
D_{KL}\!\left(f(\cdot\mid\theta_0)\,\|\,f(\cdot\mid\theta_j)\right).
\end{equation}

This quantity measures the minimal loss incurred when the true parameter value
is excluded from the model. The loss-based prior assigns higher weight to values
whose exclusion leads to larger inferential loss.

Following \citet{VillaWalker2015}, the prior is obtained by matching the
inferential loss with the self-information loss $-\log \pi(\theta)$. This leads,
up to proportionality, to
\begin{equation}\label{eq:losspriordiscrete_trim}
\pi_L(\theta)
\propto
\exp\{u(\theta)\}-1,
\qquad \theta\in\Theta.
\end{equation}

Extending this construction to continuous parameter spaces is not immediate.
If $\theta\in\Theta\subset\mathbb{R}^d$ is continuous, then for any
$\varepsilon>0$ there exists $\theta'$ arbitrarily close to $\theta$.
Under standard regularity conditions,
\[
D_{KL}\!\left(f(\cdot\mid\theta)\,\|\,f(\cdot\mid\theta')\right)\to 0
\quad \text{as } \theta'\to\theta,
\]
and therefore
\begin{equation}\label{eq:zeroKL_trim}
\inf_{\theta'\neq\theta}
D_{KL}\!\left(f(\cdot\mid\theta)\,\|\,f(\cdot\mid\theta')\right)
= 0.
\end{equation}

Therefore, in continuous parameter spaces, the exclusion of a single parameter value incurs no inferential loss. This renders the discrete approach inapplicable directly and shows that pointwise exclusion is not meaningful in this setting.

One could circumvent the issue by discretising the parameter space, which has been done, for example, in \cite{VillaWalker2014}, \cite{VillaRubio2018} or in \cite{Leisenetal2018}, but it requires a conceptual justification. Another approach would be to apply the discrete method to an increasingly fine grid, as in \cite{BrownWalker2012}; while this approach works, it obscures the role played by the geometry of the parameter space.

In the next section we address this issue by replacing point exclusion with the
exclusion of neighbourhoods, leading to a non-trivial extension of the
loss-based framework to continuous parameter spaces.

\section{Neighbourhood exclusion: The general framework}
\label{sc_neighbourhoodexclusion}

To extend the loss-based construction to continuous parameter spaces, we replace
point exclusion with the exclusion of a neighbourhood around each parameter
value.

Throughout, we assume the following regularity conditions.

\begin{assumption}[Regularity]
\label{as:regularityconditions}
Let $f(\cdot\mid\theta)$ be a statistical model with
$\theta\in\Theta\subset\mathbb{R}^d$, $d\geq1$. The model is identifiable,
$\Theta$ is open, the log-likelihood is twice continuously differentiable, and
the Fisher information matrix $I(\theta)$ exists, is finite, positive definite,
and continuous in $\theta$.
\end{assumption}

Under Assumption \ref{as:regularityconditions}, the Kullback--Leibler divergence
admits a local quadratic approximation, which ensures that small perturbations
of $\theta$ correspond to small inferential losses.

\begin{definition}[Exclusion sets]
\label{de:exclusionsets_trim}
Let $\theta\in\Theta$ and let $\delta>0$. An \emph{exclusion set} is a
neighbourhood $R_\delta(\theta)\subset\Theta$ representing the set of parameter
values removed together with $\theta$.
\end{definition}

The parameter $\delta$ controls the resolution at which parameter values are
distinguishable, while the geometry of $R_\delta(\theta)$ determines how
inferential loss is distributed across directions in the parameter space.

\begin{definition}[$\delta$-worth]
\label{de:deltaworth_trim}
Given an exclusion set $R_\delta(\theta)$, the \emph{$\delta$-worth} of $\theta$
is defined as
\begin{equation}\label{eq:deltaworth_trim}
u_\delta(\theta)
=
\inf_{\theta'\notin R_\delta(\theta)}
D_{KL}\!\left(f(\cdot\mid\theta)\,\|\,f(\cdot\mid\theta')\right).
\end{equation}
\end{definition}

The quantity $u_\delta(\theta)$ represents the minimum inferential loss incurred
when the true parameter value is approximated by the best available alternative
outside the exclusion region.

Following the loss-based principle introduced in Section
\ref{sc_background}, we match inferential loss with self-information loss, leading to the loss-based prior
\begin{equation}\label{eq:lossbasedprior_trim}
\pi_\delta(\theta)
\propto
\exp\{u_\delta(\theta)\}-1.
\end{equation}

In contrast to the discrete case, $\pi_\delta(\theta)$ is interpreted locally:
$\pi_\delta(\theta)\,d\theta$ represents the probability content over an infinitesimal
region, and only relative differences in self-information are meaningful.

The parameter $\delta$ determines the scale of neighbourhood exclusion. In what
follows we focus on the local regime $\delta\to 0$, which captures the intrinsic
geometry of inferential loss and leads to tractable expressions for the induced
prior distributions.

\section{Loss-based prior distributions}
\label{sc_lossbasedprior}

Under Assumption \ref{as:regularityconditions},  we have (see, for example, \cite{amari2007})
\begin{equation}\label{eq:quadraticexpansion_trim}
D_{KL}\!\left(f(\cdot\mid\theta)\,\|\,f(\cdot\mid\theta+h)\right)
=
\tfrac{1}{2}\,h^\top I(\theta)\,h + o(\|h\|^2),
\qquad h\to 0,
\end{equation}
where $I(\theta)$ is the Fisher information matrix.

This expansion shows that, locally, inferential loss is governed by the
quadratic form $h^\top I(\theta)h$. Directions associated with small
eigenvalues of $I(\theta)$ correspond to weak curvature, while large
eigenvalues correspond to rapid changes in the model.

\subsection{One-dimensional parameter spaces}

Let $\theta\in\Theta\subset\mathbb{R}$. The exclusion region is
\[
R_\delta(\theta)=\{\theta+h:\ |h|\le \delta\}.
\]

\begin{lemma}\label{lm:onedim_trim}
Under Assumption \ref{as:regularityconditions},
\[
u_\delta(\theta)
=
\tfrac{1}{2}I(\theta)\delta^2 + o(\delta^2),
\qquad \delta\to 0.
\]
\end{lemma}

\begin{proof}
The result follows from the quadratic expansion in \eqref{eq:quadraticexpansion_trim}, since the infimum over $|h| \leq \delta$ is attained at $|h| = \delta$.
\end{proof}

Using $\exp(u)-1 \approx u$ for small $u$:

\begin{theorem}\label{teo:onedim_trim}
The loss-based prior satisfies
\[
\pi(\theta)\propto I(\theta)^{1/2}.
\]
\end{theorem}

\begin{proof}
It follows from Lemma \ref{lm:onedim_trim},  ignoring constants independent of $\theta$, and using the approximation $\exp(u)-1\approx u$ for small $u$.
\end{proof}

Thus, in one dimension, the loss-based prior coincides with Jeffreys'
prior.

\subsection{Multidimensional parameter spaces}

Let $\theta\in\Theta\subset\mathbb{R}^d$, $d\ge2$. Define the exclusion
region
\begin{equation}
R_\delta^A(\theta)
=
\{\theta+h:\ h^\top A(\theta)h \le \delta^2\},
\label{eq:RdeltaA_trim}
\end{equation}
where $A(\theta)$ is symmetric positive definite.

\begin{theorem}\label{te:generalprior_trim}
Under Assumption \ref{as:regularityconditions},
\begin{equation}\label{eq:losstheo2}
u_\delta^A(\theta)
=
\tfrac{1}{2}\delta^2\lambda_{\min}
\bigl(A(\theta)^{-1/2}I(\theta)A(\theta)^{-1/2}\bigr)
+ o(\delta^2).
\end{equation}
Consequently,
\begin{equation}\label{eq:priortheo2}
\pi_\delta(\theta)
\propto
\lambda_{\min}
\bigl(A(\theta)^{-1/2}I(\theta)A(\theta)^{-1/2}\bigr).
\end{equation}
\end{theorem}

From the above Theorem \ref{te:generalprior_trim} (which proof is in the Appendix \ref{app:theoproof}), 
we note that the leading behaviour of the prior is governed by the smallest eigenvalue of the Fisher information
relative to the exclusion geometry. It therefore emphasises directions in
which the model is locally weakly identifiable.

\subsection{Jeffreys' prior as a special case}

When $A(\theta)=I(\theta)$, the exclusion regions are isotropic in the
Fisher geometry. In this case, the construction yields Jeffreys' prior.

\begin{corollary}
If $A(\theta)=I(\theta)$, then
\[
\pi(\theta)\propto \sqrt{\det I(\theta)}.
\]
\end{corollary}

\begin{proof}
This result reflects the fact that Fisher-isotropic exclusion corresponds to a
volume-based aggregation of local curvature, rather than a weakest-direction
criterion as in Theorem \ref{te:generalprior_trim}.
\end{proof}

Thus, Jeffreys' prior arises as the isotropic case of the loss-based
framework, while alternative choices of $A(\theta)$ lead to anisotropic
priors that emphasise directional fragility.

\section{Properties and interpretation of the loss-based construction}

The loss-based prior introduced in Section 4 depends on the choice of the exclusion
geometry $A(\theta)$. While this choice is not uniquely determined by the likelihood,
it can be guided by standard desiderata in objective Bayesian analysis. In this section
we discuss how the proposed framework accommodates these principles, and how different
choices of $A(\theta)$ lead to distinct but interpretable prior constructions. A more thorough and detailed illustration of the examples can be found in Appendix B.

\subsection{Invariance under reparametrisation}

A widespread desideratum is invariance under smooth one-to-one reparametrisations
$\phi = \phi(\theta)$. In the present framework, invariance is achieved by requiring
that $A(\theta)$ transforms as a $(0,2)$ tensor:
\[
A_{\phi}(\phi) = J(\theta)^{-T} A_{\theta}(\theta) J(\theta)^{-1},
\]
where $J(\theta)$ is the Jacobian of the transformation. Under this condition,
the matrix $A(\theta)^{-1/2} I(\theta) A(\theta)^{-1/2}$ transforms by similarity,
and therefore its eigenvalues, and hence the resulting prior, are invariant.

A canonical choice is $A(\theta) = I(\theta)$, corresponding to isotropic exclusion
in the Fisher geometry. This yields Jeffreys' prior. More generally, invariant
anisotropic geometries can be constructed as
\[
A(\theta) = I(\theta)^{1/2} B(\theta) I(\theta)^{1/2},
\]
where $B(\theta)$ is a positive definite matrix defined in Fisher units.

\subsection{Likelihood principle and independence of the realised data}

Objective priors are typically required to be independent of the realised data.
Within the loss-based framework, this property is ensured whenever $A(\theta)$
depends only on the model and not on the observed sample.

In regression settings, it is natural to allow dependence on the design matrix $X$,
which is part of the data-generating mechanism. For example, in logistic regression,
a choice such as $A(\beta) = X^\top X$ yields a prior independence of the realised data. By contrast,
choices involving estimated quantities, such as $A(\beta) = X^\top W(\hat{\beta}) X$,
introduce dependence on the observed data and therefore violate the likelihood principle.

\subsection{Group invariance}

In models admitting a group structure, it is often desirable for the prior to be
invariant under transformations preserving the statistical model. Within the present
framework, this requirement can be enforced by choosing $A(\theta)$ to transform
equivariantly under the group action, so that the exclusion regions are preserved.

Fisher-isotropic exclusion, $A(\theta) = I(\theta)$, automatically satisfies this
property and recovers Haar-type priors in classical settings. For example, in the
normal location--scale model, this leads to $\pi(\mu,\sigma) \propto 1/\sigma$.
More generally, anisotropic choices of $A(\theta)$ allow controlled departures from
group invariance while preserving the underlying geometric structure.

\subsection{Interest--nuisance structure}

In multiparameter problems, it is often necessary to distinguish between parameters
of interest and nuisance components. In the present framework, this asymmetry can be
encoded directly through the geometry of the exclusion region.

Let $\theta = (\psi,\lambda)$, where $\psi$ denotes the parameter of interest.
A block-structured choice of $A(\theta)$,
\[
A(\theta) =
\begin{pmatrix}
A_{\psi\psi}(\theta) & 0 \\
0 & A_{\lambda\lambda}(\theta)
\end{pmatrix},
\]
allows different directions of the parameter space to be weighted differently.
Directions associated with larger weights correspond to exclusions that incur
greater inferential loss, and are therefore more strongly protected.

This provides a geometric alternative to reference-prior constructions, in which
inferential priorities are encoded directly through the exclusion metric rather
than through an optimisation of missing information.

\subsection{Stability and weak identification}

Frequentist matching properties provide one important criterion for evaluating
objective priors, but they primarily reflect average behaviour under repeated
sampling. They do not directly address situations in which the likelihood is nearly
flat or the Fisher information becomes ill-conditioned.

The loss-based construction offers a complementary perspective by focusing on local
inferential stability. In particular, priors driven by the minimum eigenvalue of
$A(\theta)^{-1/2} I(\theta) A(\theta)^{-1/2}$ assign reduced weight to regions where
the model is weakly identified. This penalises directions of low curvature and
mitigates the impact of near-singular configurations on posterior inference.

For instance, in models where parameters become nearly non-identifiable, such as
correlation parameters approaching boundary values, the proposed priors tend to
downweight these regions relative to volume-based constructions. This reflects a
design principle based on protection against local inferential fragility rather
than global calibration.

\section{Discussion}

We have proposed a general formulation of loss-based priors for continuous parameter spaces,
extending the framework originally developed for discrete settings. The key idea is to replace
point exclusion with neighbourhood exclusion, thereby restoring a non-trivial notion of
inferential loss and linking prior construction to the local geometry induced by the
Kullback--Leibler divergence.

The resulting class of priors is indexed by the choice of exclusion geometry $A(\theta)$, which
plays a role analogous to that of a loss function in decision theory. This perspective provides
a unified framework in which several well-known objective priors arise as special cases. In
particular, Fisher-isotropic exclusion recovers Jeffreys' prior, while more general choices of
$A(\theta)$ lead to anisotropic constructions that reflect directional features of the model.

A central aspect of the proposed approach is that it separates the modelling of inferential loss
from the likelihood itself. This allows standard desiderata of objective Bayesian analysis—such
as invariance, data-independence, and group symmetry—to be incorporated through the choice of
$A(\theta)$, while also enabling alternative principles, such as robustness to weak
identification, to be expressed in a transparent and interpretable way.

The normal model provides a useful benchmark for interpreting the behaviour of the proposed
priors. In this setting, the likelihood is regular, the parameters are orthogonal, and the
Fisher geometry is stable. As a result, Jeffreys' and reference priors already exhibit strong
frequentist and invariance properties, and are often regarded as canonical choices. It is
therefore not surprising that anisotropic, fragility-driven priors do not improve standard
frequentist criteria such as mean squared error or credible set coverage in this case. The
normal model should thus be viewed primarily as an illustration of the framework, rather than
as a setting in which performance gains are expected.

More generally, the proposed construction is motivated by a different notion of inferential
risk. While classical criteria such as frequentist matching focus on average behaviour under
repeated sampling, they do not directly address situations in which the likelihood becomes
nearly flat or the Fisher information is ill-conditioned. In such cases, inference may be
sensitive to small perturbations of the parameter, leading to weak identification or numerical
instability. The loss-based framework addresses this issue by assigning lower weight to regions
where the model is locally fragile, as captured by the weakest directions of the Fisher
information relative to the exclusion geometry.

From this perspective, anisotropic loss-based priors should be interpreted as stability-oriented
rather than optimal under standard frequentist loss functions. Their role is not to dominate
Jeffreys' or reference priors in well-behaved models, but to provide protection in settings
where inferential geometry is highly anisotropic or near-singular. In such regimes, volume-based
aggregation can mask weak directions, whereas fragility-based constructions explicitly respond
to them.

The framework does not prescribe a unique prior, but rather defines a principled class of
candidates. The choice of exclusion geometry may therefore be guided by the inferential goals of
the analysis, reflecting a balance between invariance, interpretability, and stability. In this
sense, the proposed construction should be viewed as complementary to existing objective
Bayesian methods, rather than as a replacement.

Several directions for future work remain. These include the development of principled criteria
for selecting $A(\theta)$ in specific model classes, a systematic comparison with reference and
matching priors in weakly identified settings, and the extension of the framework to models
with singularities or boundary effects. More broadly, the connection between loss-based priors
and information geometry suggests further links with geometric and decision-theoretic
approaches to Bayesian inference.

\bibliography{references}

\appendix
\appendix
\section{Proof of Theorem \ref{te:generalprior_trim}}\label{app:theoproof}

Fix $\theta\in\Theta$ and write
\[
A=A(\theta), \qquad I=I(\theta),
\]
for brevity. Since $\Theta$ is open, there exists $\varepsilon_0>0$ such that
$\theta+h\in\Theta$ whenever $\|h\|<\varepsilon_0$.

Because $A$ is symmetric positive definite, there exist constants
$0<m_A\le M_A<\infty$ such that
\begin{equation}
m_A\|h\|^2 \le h^\top A h \le M_A\|h\|^2,
\qquad h\in\mathbb{R}^d.
\label{eq:A_bounds_app}
\end{equation}
Hence, if $h^\top A h=\delta^2$, then
\[
\|h\| \le \frac{\delta}{\sqrt{m_A}} \to 0
\qquad \text{as } \delta\to 0.
\]

Let
\[
g(h)
:=
D_{KL}\!\left(f(\cdot\mid\theta)\,\|\,f(\cdot\mid\theta+h)\right).
\]
By the quadratic expansion \eqref{eq:quadraticexpansion_trim},
\[
g(h)=\frac12 h^\top I h + r_\theta(h),
\qquad r_\theta(h)=o(\|h\|^2),
\qquad h\to 0.
\]

We wish to evaluate
\[
u_\delta^A(\theta)
=
\inf_{h^\top A h\ge \delta^2} g(h).
\]

\medskip

\noindent\textbf{Step 1: reduction to the boundary.}
Since $I$ is positive definite, the quadratic term $h^\top I h$ is strictly positive for
$h\neq 0$ and homogeneous of degree two. Therefore, for $\delta$ sufficiently small, the
infimum of $g(h)$ over the set $\{h:\,h^\top A h\ge \delta^2\}$ is attained asymptotically
at the smallest admissible perturbations, namely on the boundary
\[
\mathcal S_\delta:=\{h\in\mathbb{R}^d:\ h^\top A h=\delta^2\}.
\]
Thus,
\begin{equation}
u_\delta^A(\theta)
=
\inf_{h\in\mathcal S_\delta} g(h) + o(\delta^2).
\label{eq:boundary_reduction_app}
\end{equation}

\medskip

\noindent\textbf{Step 2: uniform control of the remainder.}
For $h\in\mathcal S_\delta$, \eqref{eq:A_bounds_app} gives
\[
\|h\| = O(\delta)
\qquad \text{uniformly on } \mathcal S_\delta.
\]
Hence
\[
\sup_{h\in\mathcal S_\delta}|r_\theta(h)| = o(\delta^2).
\]
It follows that
\begin{equation}
\inf_{h\in\mathcal S_\delta} g(h)
=
\inf_{h\in\mathcal S_\delta}\frac12 h^\top I h + o(\delta^2).
\label{eq:remainder_split_app}
\end{equation}

\medskip

\noindent\textbf{Step 3: quadratic minimisation.}
We now minimise $h^\top I h$ subject to $h^\top A h=\delta^2$.
Set
\[
v=A^{1/2}h,
\qquad \text{so that} \qquad
h=A^{-1/2}v.
\]
Then the constraint becomes
\[
\|v\|^2=\delta^2,
\]
and
\[
h^\top I h
=
v^\top\!\left(A^{-1/2} I A^{-1/2}\right)v.
\]
Therefore
\[
\inf_{h\in\mathcal S_\delta} h^\top I h
=
\inf_{\|v\|=\delta} v^\top\!\left(A^{-1/2} I A^{-1/2}\right)v.
\]
By the Rayleigh--Ritz theorem,
\[
\inf_{\|v\|=\delta} v^\top\!\left(A^{-1/2} I A^{-1/2}\right)v
=
\delta^2 \lambda_{\min}\!\left(A^{-1/2} I A^{-1/2}\right).
\]
Hence
\begin{equation}
\inf_{h\in\mathcal S_\delta}\frac12 h^\top I h
=
\frac12 \delta^2
\lambda_{\min}\!\left(A^{-1/2} I A^{-1/2}\right).
\label{eq:quadratic_min_app}
\end{equation}

\medskip

\noindent\textbf{Step 4: asymptotic form of the $\delta$-worth.}
Combining \eqref{eq:boundary_reduction_app}, \eqref{eq:remainder_split_app},
and \eqref{eq:quadratic_min_app}, we obtain
\[
u_\delta^A(\theta)
=
\frac12 \delta^2
\lambda_{\min}\!\left(A^{-1/2} I A^{-1/2}\right)
+ o(\delta^2),
\]
which proves \eqref{eq:losstheo2}.

\medskip

\noindent\textbf{Step 5: induced prior.}
From the loss-based construction,
\[
\pi_\delta(\theta)\propto \exp\{u_\delta^A(\theta)\}-1.
\]
Since $u_\delta^A(\theta)=O(\delta^2)$, we have
\[
\exp\{u_\delta^A(\theta)\}-1
=
u_\delta^A(\theta)\{1+o(1)\}.
\]
Thus, up to multiplicative constants independent of $\theta$,
\[
\pi_\delta(\theta)
\propto
\lambda_{\min}\!\left(A(\theta)^{-1/2} I(\theta) A(\theta)^{-1/2}\right),
\]
which proves \eqref{eq:priortheo2}.
\qed

\section{Appendix B: Illustrative examples for the desiderata}\label{app:illustrations}

In this appendix we provide simple examples illustrating how different choices
of the exclusion geometry $A(\theta)$ support the desiderata discussed in Section~5. Each subsection corresponds directly to the desiderata introduced in Section 5, using simple models to illustrate the role of the exclusion geometry.

\subsection{D1: Invariance under reparametrisation}

Consider the normal model $X \sim \mathcal{N}(\mu,\sigma^2)$ and the
reparametrisation $\tau = \sigma^{-2}$.

The Fisher information matrices are
\[
I(\mu,\sigma^2) =
\begin{pmatrix}
\sigma^{-2} & 0 \\
0 & (2\sigma^4)^{-1}
\end{pmatrix}, \qquad
I(\mu,\tau) =
\begin{pmatrix}
\tau & 0 \\
0 & (2\tau^2)^{-1}
\end{pmatrix}.
\]

Choosing Fisher-isotropic exclusion $A(\theta)=I(\theta)$ yields
\[
\pi(\mu,\sigma^2) \propto \sigma^{-3/2}, \qquad
\pi(\mu,\tau) \propto \tau^{-1/2}.
\]

These are consistent under the change of variables, confirming invariance.

\medskip

\noindent\textit{Interpretation.}
Tensorial transformation of $A(\theta)$ ensures that the induced prior is invariant.

\subsection{D2: Likelihood principle and data-independence}

Consider logistic regression
\[
Y_i \sim \text{Bernoulli}(p_i), \qquad \logit(p_i) = x_i^\top \beta.
\]

The Fisher information is
\[
I(\beta) = X^\top W(\beta) X.
\]

Two choices of exclusion geometry:

\begin{itemize}
\item Model-based:
\[
A(\beta) = X^\top X,
\]
which depends only on the design matrix.

\item Data-dependent:
\[
A(\beta) = X^\top W(\hat\beta) X,
\]
which depends on the observed data.
\end{itemize}

The corresponding priors are
\[
\pi(\beta) \propto \lambda_{\min}\!\left(A^{-1/2} I(\beta) A^{-1/2}\right).
\]

\medskip

\noindent\textit{Interpretation.}
Only the first choice respects the likelihood principle, since it does not depend on the realised data.

\subsection{D3: Group invariance}

Consider the normal location–scale model
\[
X \sim \mathcal{N}(\mu,\sigma^2),
\]
which is invariant under
\[
x \mapsto ax + b, \qquad \mu \mapsto a\mu + b, \qquad \sigma \mapsto a\sigma.
\]

Choosing
\[
A(\mu,\sigma) = I(\mu,\sigma),
\]
yields
\[
\pi(\mu,\sigma) \propto \frac{1}{\sigma},
\]
which coincides with the Haar prior.

\medskip

\noindent\textit{Interpretation.}
Fisher-isotropic exclusion automatically preserves group symmetries.

\subsection{D4: Interest--nuisance structure}

Let
\[
X_i \sim \mathcal{N}(\mu,\sigma^2), \qquad \theta = (\mu,\sigma^2),
\]
and suppose $\mu$ is of interest and $\sigma^2$ is nuisance.

Consider a block geometry
\[
A(\theta) =
\begin{pmatrix}
1 & 0 \\
0 & c
\end{pmatrix}, \qquad c \gg 1.
\]

Then
\[
\pi(\theta) \propto \lambda_{\min}\!\left(A^{-1/2} I(\theta) A^{-1/2}\right)
= \min\{\sigma^{-2},\, c^{-1}(2\sigma^4)^{-1}\}.
\]

\medskip

\noindent\textit{Interpretation.}
Large $c$ increases protection in the nuisance direction, mimicking the behaviour of reference priors.

\subsection{D5: Stability and weak identification}

Consider a bivariate normal model with correlation $\rho$:
\[
(X_1,X_2) \sim \mathcal{N}(0,\Sigma), \qquad
\Sigma =
\begin{pmatrix}
1 & \rho \\
\rho & 1
\end{pmatrix}.
\]

The Fisher information for $\rho$ behaves as
\[
I(\rho) \propto \frac{1}{(1-\rho^2)^2},
\]
which becomes ill-conditioned as $|\rho|\to 1$.

Choosing Euclidean exclusion $A=1$ gives
\[
\pi(\rho) \propto I(\rho),
\]
which diverges at the boundary.

By contrast, a fragility-oriented construction yields
\[
\pi(\rho) \propto \lambda_{\min}(I(\rho)),
\]
which downweights regions where curvature becomes unstable.  In one dimension, the minimum-eigenvalue construction coincides with the Fisher
information itself. However, the interpretation remains that the prior reflects
local curvature and penalises regions where inferential stability deteriorates.

\medskip

\noindent\textit{Interpretation.}
The loss-based prior penalises weakly identified regions, improving stability relative to volume-based priors.

\end{document}